\documentclass[11pt,a4paper]{article}
 
\usepackage[margin=1.4cm]{geometry}
\usepackage{amsmath,amssymb,amsthm}
\usepackage{graphicx}
\usepackage{booktabs}
\usepackage{tikz}
\usetikzlibrary{arrows.meta,positioning,calc}
\usepackage[colorlinks=true,linkcolor=blue!55!black,citecolor=blue!55!black,
            urlcolor=blue!55!black]{hyperref}
\usepackage{caption}
\usepackage{authblk}
\usepackage{enumitem}

\newtheorem{theorem}{Theorem}[section]
\newtheorem{proposition}[theorem]{Proposition}
\newtheorem{lemma}[theorem]{Lemma}
\newtheorem{corollary}[theorem]{Corollary}
\theoremstyle{definition}
\newtheorem{definition}[theorem]{Definition}
\newtheorem{remark}[theorem]{Remark}
 
\newcommand{\C}{\mathbb{C}}

\newcommand{\E}{\mathbb{E}}
\newcommand{\Var}{\operatorname{Var}}
\newcommand{\Rey}{\operatorname{Re}}
\DeclareMathOperator{\sinc}{sinc}
\newcommand{\Uset}{\mathcal{U}}
\newcommand{\Vset}{\mathcal{V}}

\title{\bfseries A Gauge-Invariant Clustering Coefficient for Complex-Weighted Bipartite Networks}
 
\author[1]{Anjali Waghmare}
\author[2]{Yu Tian}
\author[1]{Renaud Lambiotte}
\affil[1]{\textit{Mathematical Institute, University of Oxford, OX2 6GG Oxford, UK}}
\affil[2]{\textit{Center for Systems Biology Dresden, 01307 Dresden, Germany}}
\date{}

\begin{document}
\maketitle

\begin{abstract}
\noindent
Structural measures such as the clustering coefficient and the average shortest-path length characterise how a network is organised. These measures are typically formulated for real, non-negative edge weights. A class of quantum and photonic architectures has complex edge weights instead, whose phases determine whether alternative routes interfere constructively or destructively. These architectures are also bipartite, so triangles are absent and the smallest closed cycle is a four-node square. Existing measures address complex weights and bipartite structure separately: bipartite clustering coefficients quantify clustering through four-node squares but do not contain phase information, while interferometric coefficients retain phase but are defined on triangles. In this work, we define a clustering coefficient for complex-weighted bipartite networks which reduces to the classical bipartite coefficient when the phases vanish, becomes negative when alternative routes cancel, and can be obtained for all nodes from a single sparse matrix product. We show that the phase accumulated around a square is the smallest gauge-invariant carrier of structural phase information in a bipartite network. The clustering coefficient factorises into a topological contribution and a phase contribution, making a bipartite Watts--Strogatz ensemble analytically tractable. For phases uniformly distributed on $[-\Delta,\Delta]$, the mean phase contribution is $(\sin\Delta/\Delta)^4$, independent of node, degree, and topology. We also derive closed-form expressions for the clustering coefficient, open-path visibility, and phase variance. We numerically verify these predictions and further show that phase disorder increases coherent distance.
\end{abstract}

\tableofcontents

\section{Introduction}
The structure of a network affects the way information moves through it, and a small set of measures is often used to describe this structure. Among the most prominently used measures are the clustering coefficient, which measures the extent to which the neighbourhood of a node is interconnected, and the average shortest-path length, which is the mean number of links on the shortest path between two nodes. Clustering coefficients have been used to quantify community formation in social networks~\cite{granovetter1973,newman2003social}, redundancy and fault tolerance in engineered infrastructure~\cite{crucitti2004,albert2004}, and functional modules in biological interaction networks~\cite{ravasz2002,barabasi2004}. The average shortest-path length has been used to study how efficiently information, influence or disease can spread across a network, as in studies of social distance~\cite{travers1969,watts1998}, transport efficiency~\cite{latora2001} and epidemic spreading~\cite{pastor2001}. The joint behaviour of such measures helps distinguish a regular lattice from a random graph and identifies the small-world regime lying between them~\cite{watts1998}. These measures are typically formulated for edge weights that are real and non-negative, and the case of complex-valued edge weights remains relatively less explored.
 
However, a class of engineered quantum and photonic systems falls outside this setting. In a quantum repeater, a superconducting processor or a programmable photonic mesh, an edge carries a phase in addition to a strength. As a result, signals propagate as complex amplitudes, and two routes connecting the same pair of components may cancel one another. Weights of this form also describe quantum amplitudes~\cite{nielsen2010}, wave propagation~\cite{bracewell2000}, and oscillators coupled with a phase lag~\cite{arenas2008}. The performance of such a device frequently depends on interference between alternative routes. In quantum repeaters, this interference is determined by the relative phase of the optical amplitudes arriving at two quantum memories~\cite{briegel1998,sangouard2011}; in superconducting processors, by the phase accumulated around loops formed by designed links and residual crosstalk~\cite{blais2021,kurpiers2018}; and in programmable photonic meshes, by the phases of the interferometers that compose the network~\cite{reck1994,clements2016,bogaerts2020}. These phase-dependent interference effects have also motivated the development of network measures for complex-valued weights. These include extensions of standard network measures to complex weights~\cite{bottcher2024}, classifications of complex-weighted networks based on structural balance together with corresponding random-walk analyses~\cite{tian2024}, and measures of interference between competing paths that demonstrate its effects on small-world behaviour~\cite{krawciw2024}. 

In addition to their complex-valued edge weights, a second defining feature of these systems is their bipartite structure~\cite{latapy2008}. The nodes of a bipartite network are partitioned into two classes, with links permitted only between nodes belonging to different classes. In the systems considered here, these classes may correspond to memories and photonic channels, qubits and resonators, or optical modes and couplers. As a result of the bipartite structure, every closed path must alternate between the two classes, thus bipartite networks contain no odd-length cycles. Their shortest loops are therefore four-node squares instead of triangles, in which two nodes of the same class are connected through two distinct intermediaries of the other class (Figure~\ref{fig:systems}). Because the standard clustering coefficient quantifies local closure through triangles, it vanishes identically on bipartite networks. The natural analogue is therefore to quantify local closure through four-cycles: two nodes of the same class connected through two distinct intermediaries form a square, providing a measure of the redundancy of routes through the opposite class. Several square-based clustering coefficients have been introduced for bipartite networks~\cite{robins2004,lind2005,zhang2008,opsahl2013}.

\begin{figure}[t]
\centering
\includegraphics[width=0.8\textwidth]{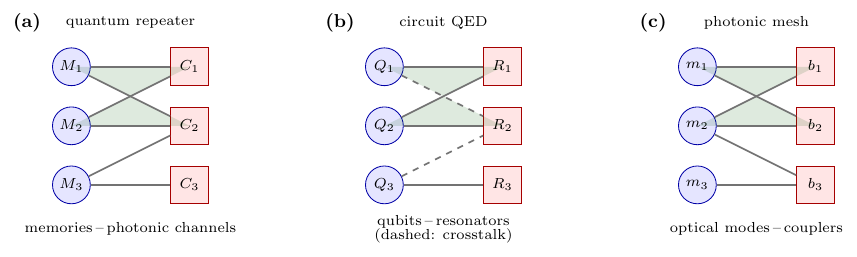}
\caption{Architectures that are bipartite as well as complex weighted. In each panel, circles denote nodes of one class and squares nodes of the other, and every link carries a phase in addition to a strength. The shaded quadrilateral marks a four-node square, the shortest closed path available, in which two routes join the same pair of same-class nodes through distinct intermediaries.
\textbf{(a)} In a quantum repeater, memories are coupled to photonic channels, and the success of entanglement swapping between two memories depends on the relative phase of the optical amplitudes arriving along the two routes~\cite{briegel1998,sangouard2011}. \textbf{(b)} In circuit quantum electrodynamics (QED), transmon qubits are coupled to microwave resonators. The designed links form a bipartite network, and residual crosstalk (dashed) closes it into squares whose accumulated phase enters the coherent error budget of a gate~\cite{blais2021,kurpiers2018}. \textbf{(c)} In a programmable photonic mesh, optical modes are coupled by beam splitters with tunable phase shifters, and each elementary square constitutes a Mach--Zehnder interferometer~\cite{reck1994,clements2016,bogaerts2020}. In all three cases, the quantity of physical interest is the phase accumulated around the square.}
\label{fig:systems}
\end{figure}
 
However, a measure that accounts for both the bipartite structure and complex weights is largely missing from the current literature: bipartite clustering coefficients capture four-cycle structure but do not contain phase information, whereas interferometric
measures retain phase information but are formulated on triangles, which are absent in bipartite
networks. This brings us to the question:

\begin{quote}
    \textit{How should local structure be quantified in a network that is both bipartite and complex weighted?}
\end{quote}

To close this gap, we introduce a gauge-invariant clustering coefficient for complex-weighted bipartite networks. The proposed coefficient is based on the phase accumulated around a four-node cycle, which is invariant under node-wise phase redefinitions and is therefore independent of the choice of phase convention. In the absence of phase differences, the coefficient reduces to the classical bipartite clustering coefficient. Its sign further captures the distinction between constructive and destructive interference, and its exact decomposition separates the effects of network connectivity from those of the phase structure. Our main contributions are as follows:
 
\begin{enumerate}\itemsep3pt
\item \textbf{We define a clustering coefficient for complex-weighted bipartite networks.} It is built from the real parts of square amplitudes, normalised by the strength of the three-link paths that could close into squares. It reduces to the classical bipartite clustering coefficient when the phases vanish, can become negative when alternative routes interfere destructively, and is invariant under redefinition of the phase references. The coefficient is computable for all nodes from a single sparse matrix product.
 
\item \textbf{We show that the coefficient factorises exactly.} It is the product of a factor
determined by the magnitudes and the connectivity and a factor determined by the phases. This separates the structural and phase-dependent components and makes the coefficient analytically tractable. The two effects can thus be measured, varied and interpreted separately.
 
\item \textbf{We derive closed-form results for a bipartite Watts--Strogatz ensemble.} For a ring lattice randomised with probability $p$ and with link phases drawn uniformly from $[-\Delta,\Delta]$, the mean phase contribution is $(\sin\Delta/\Delta)^4$, independent of node, degree, and topology. We also obtain closed-form expressions for the clustering coefficient, the visibility of open paths, and the variance of the phase contribution.
 
\item \textbf{We show how coherence enters the global structure.} The topological contribution is obtained numerically, as its dependence on $p$ does not admit a simple closed form. Phase disorder leaves hop-count distances unchanged but increases the distance experienced by coherent signals, so that combinatorial and coherent notions of distance no longer coincide. We also show that the phase contribution cancels from the small-world coefficient when a network is compared with a randomised version of itself, making its observed phase dependence a consequence of the choice of reference network.

\end{enumerate}
 
The remainder of the paper is organised as follows. In Section~\ref{sec:lit}, we review prior work and position our work within the literature. In Section~\ref{sec:framework}, we introduce the notation, define the clustering measures, and establish the gauge invariance of the proposed construction. In Section~\ref{sec:theory}, we derive their analytical properties, and in Section~\ref{sec:numerics}, we verify their behaviour numerically. In Section~\ref{sec:conclusion}, we discuss the implications of the theoretical and numerical results. Proofs are provided in Appendix~\ref{app:proofs}, the connections with existing measures are established in Appendix~\ref{app:relations}, and computational details are provided in Appendix~\ref{app:computation}.
 
\section{Related Literature}\label{sec:lit}

\paragraph{Clustering in bipartite networks.} Work on clustering in bipartite networks proceeds from the observation that the triangle count vanishes identically and replaces it by a count of four-cycles. Early approaches measured global cohesion through the ratio of four-cycles to three-link paths, yielding a statistic analogous to transitivity~\cite{robins2004}. Local coefficients were subsequently introduced based on the fraction of neighbour pairs that share additional common neighbours, with comparisons showing behaviour closely related to the triangle coefficient across the Watts--Strogatz transition~\cite{lind2005}. Later work defined node-level and link-level square coefficients and applied them to community detection, while also revisiting the distinction between global and local measures and the treatment of weights~\cite{zhang2008}. For weighted networks, a common approach averages the closing link using the strengths of the paths reaching it~\cite{zhanghorvath2005}, and alternative formulations have also been compared~\cite{saramaki2007}. All of these measures operate on real, non-negative weights and do not account for a phase.

\paragraph{Complex-weighted networks.} A separate line of work has extended network analysis to complex-valued edge weights. Broad treatments generalise a range of standard network measures and catalogue the domains in which Hermitian weight matrices arise~\cite{bottcher2024}. Other work classifies Hermitian complex-weighted networks as balanced, antibalanced, or strictly unbalanced, relates these classes to the spectrum, and analyses random walks and spectral clustering, with the magnetic Laplacian operator arising as a special case~\cite{tian2024}. This operator has an extensive history in mathematical physics~\cite{shubin1994,lange2015} and has also been applied to community detection and representation learning on directed graphs~\cite{fanuel2017,zhang2021magnet}. Of particular relevance here, the connection between structural balance and gauge invariance has been established through Cheeger inequalities for magnetic Laplacians~\cite{lange2015}. This correspondence is used in the gauge-invariance argument we develop in Section~\ref{sec:gauge}.

\paragraph{Interferometric measures.} Closest to our construction is the interferometric family introduced by Krawciw, Carr and Diniz Behn~\cite{krawciw2024} for directed complex-weighted networks. Their clustering coefficient compares the amplitude $W_{ki}W_{ij}$ of a two-step path through an intermediary with the amplitude $W_{kj}$ of the direct link that would close the triangle,
\begin{equation}
  C^{\mathrm{int}}_i=
  \frac{\sum_{j\neq k}|W_{ki}W_{ij}|
        \bigl(|W_{kj}+W_{ki}W_{ij}|-|W_{ki}W_{ij}|\bigr)}
       {\sum_{j\neq k}|W_{ki}W_{ij}|},
  \label{eq:krawciw}
\end{equation}

extending the framework in Ref.~\cite{zhanghorvath2005} and admitting negative values when the two amplitudes interfere destructively. They further define an apparent path length from the resolvent of the weight matrix, which accounts for walks of different lengths, and report that constructive interference strengthens the small-world effect while destructive interference weakens it. The coefficient defined in Section~\ref{sec:clustering} is the bipartite counterpart of Equation~\eqref{eq:krawciw}, with two two-step paths through distinct intermediaries replacing a path and a link, and we adopt their apparent path length without modification.
 
\paragraph{Small-worldness and coherent transport.} Two further aspects of the literature are relevant to interpreting our results. The first concerns the quantification of small-worldness. The coefficient introduced by Humphries and Gurney~\cite{humphries2008} compares clustering and path length with those of a randomised reference network. Because the reference clustering appears in the denominator and is typically very small, the coefficient can be highly sensitive to slight changes
in it~\cite{telesford2011}. As shown in Section~\ref{sec:sw}, complex weights make this issue more pronounced, since with fully disordered phases the reference clustering has mean zero. The second aspect concerns coherent transport. Continuous-time quantum walks provide a standard model of amplitude propagation on graphs~\cite{mulken2011}, and photonic implementations are naturally realised as bipartite lattices of modes and couplers~\cite{peruzzo2010}. In such systems, destructive interference can suppress a transfer that would occur under classical dynamics~\cite{mohseni2008}. Phase disorder can likewise be represented by assigning random phases to the links of a lattice, as in the random-flux model of localisation theory~\cite{furusaki1999,altland1999}. In this framework, the physically relevant quantity is the distribution of fluxes, which is consistent with the invariance imposed here.

No existing measure is simultaneously bipartite, phase-sensitive, and invariant under redefinition of the phase references. We fill this gap by constructing a coefficient built on the square flux, which satisfies the requirements simultaneously.

\section{Complex-Weighted Bipartite Clustering Coefficients}\label{sec:framework}

We begin by fixing the notation for complex-weighted bipartite networks and showing that the four-node square carries the phase information. We then define the clustering coefficient together with the open-path visibility and the coherent distance, before describing the ensemble on which the results are obtained.
 
\subsection{Complex-Weighted Bipartite Networks}
 
A bipartite network is a triple $G=(\Uset,\Vset,E)$ in which $\Uset$ and $\Vset$ are disjoint sets of nodes, of sizes $n_U$ and $n_V$, and $E\subseteq \Uset\times\Vset$ is the set of links. Each link carries a complex weight
\begin{equation}
  w_{uv}=|w_{uv}|\,e^{i\varphi_{uv}},
  \qquad |w_{uv}|>0,\quad \varphi_{uv}\in[-\pi,\pi),
\end{equation}
whose modulus $|w_{uv}|$ is the strength and whose argument $\varphi_{uv}$ is the phase acquired in passing from $u$ to $v$. We impose $w_{vu}=\overline{w_{uv}}$, so that traversing a link and returning acquires no net phase. The weight matrix $A$ is then Hermitian,
\begin{equation}
  A=\begin{pmatrix}0 & B\\ B^{\mathsf H} & 0\end{pmatrix},
  \qquad B\in\C^{n_U\times n_V},\quad B_{uv}=w_{uv},
  \label{eq:adjacency}
\end{equation}
where the diagonal blocks vanish because links join only nodes of different classes. Hermiticity generalises the symmetry of a real-weighted undirected network and expresses reciprocity in wave systems and time-reversal symmetry in quantum mechanics~\cite{bottcher2024,tian2024}. All structural information resides in the biadjacency matrix $B$, which is the object we compute with. The amplitude of a path or cycle $P$ is the ordered product of its link weights, $W(P)=\prod_{(x,y)\in P}w_{xy}$. We write $R=|B|$ for the elementwise modulus of $B$ and $R^{\circ2}$ for its elementwise  \emph{square}, $\Delta$ for the half-width of the phase distribution, $N(u)$ for the set of neighbours of node $u$, and $q_{uu'}$ for the number of neighbours shared by two nodes $u$ and $u'$ of the same class. We denote averages over nodes by an overbar.
 
\subsection{The Square as the Elementary Unit of Interference}\label{sec:gauge}
 
Since the shortest cycles have length four, squares constitute the fundamental motif, and we define them together with the three-link paths that could close into one.
 
\begin{definition}[Ordered squares and three-link paths]\label{def:sets}
For $u\in\Uset$ let
\[
  S(u)=\bigl\{(v_1,u',v_2):v_1\neq v_2\in\Vset,\ u'\neq u\in\Uset,\
  (u,v_1),(u',v_1),(u',v_2),(u,v_2)\in E\bigr\}
\]
denote the ordered squares incident to $u$, and let
\[
  T(u)=\bigl\{(v_1,u',v_2):v_1\neq v_2,\ u'\neq u,\
  (u,v_1),(u',v_1),(u',v_2)\in E\bigr\}\supseteq S(u)
\]
denote the ordered three-link paths leaving $u$, for which the closing link $(u,v_2)$ need not exist. Their amplitudes are $W_s=w_{uv_1}w_{v_1u'}w_{u'v_2}w_{v_2u}$ for a square and $W_t=w_{uv_1}w_{v_1u'}w_{u'v_2}$ for a path, and $m_u$ denotes the number of ordered squares incident to a node $u$, thus $m_u=|S(u)|$. The corresponding definitions for $v\in\Vset$ follow on replacing $B$ by $B^{\mathsf H}$.
\end{definition}
 
Each geometric square occurs twice in $S(u)$, and both orderings are retained. Their amplitudes are conjugate, so the coherent sum over $S(u)$ is real (Proposition~\ref{prop:props}\,(vi)). In terms of the biadjacency matrix,
\begin{equation}
  W_s=B_{uv_1}\overline{B_{u'v_1}}B_{u'v_2}\overline{B_{uv_2}},
  \qquad
  \Phi_s:=\arg W_s =\varphi_{uv_1}-\varphi_{u'v_1}+\varphi_{u'v_2}-\varphi_{uv_2},
  \label{eq:Ws}
\end{equation}
and $\Phi_s$ is termed the flux of the square. In the terminology of lattice gauge theory it is the flux through a plaquette~\cite{kogut1979}, and it is the quantity entering a tight-binding Hamiltonian through the Peierls substitution~\cite{shubin1994}. The alternating signs follow from the Hermitian convention, two of the four steps running against the orientation in which their phase was defined.
 
The flux is invariant under a change of phase reference. Writing $\theta_u$ for the shift applied at $u\in\Uset$ and $\eta_v$ for the shift at $v\in\Vset$, the substitution $\varphi_{uv}\mapsto\varphi_{uv}+\theta_u+\eta_v$ corresponds to $B\mapsto D_\theta BD_\eta$ with $D_\theta=\mathrm{diag}(e^{i\theta_u})$ and $D_\eta=\mathrm{diag}(e^{i\eta_v})$, and constitutes a gauge transformation.
 
\begin{proposition}[Invariance under a change of phase reference]\label{prop:gauge}
The amplitude $W_s$, and consequently both $\Phi_s$ and $|W_s|$, are unchanged by every gauge transformation.
\end{proposition}
 
The proof, given in Appendix~\ref{app:proofs}, consists of a single cancellation, each of $\theta_u$, $\theta_{u'}$, $\eta_{v_1}$ and $\eta_{v_2}$ entering~\eqref{eq:Ws} once with each sign. Figure~\ref{fig:schematic} illustrates the structure and its two readings, as a plaquette and as a two-slit interferometer.
 
\begin{figure}[t]
\centering
\includegraphics[width=0.9\textwidth]{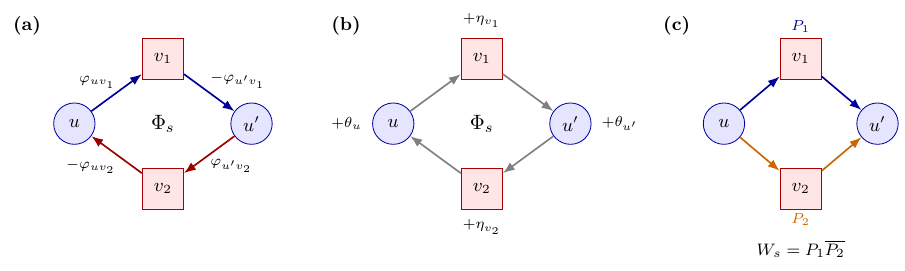}
\caption{Circles denote nodes of $\Uset$ and squares nodes of $\Vset$,
and $\varphi_{uv}$ is the phase of the link between $u$ and $v$. \textbf{(a)} The flux $\Phi_s$ accumulated around a square is an alternating sum of the four link phases, since under the convention $w_{vu}=\overline{w_{uv}}$ two of the four steps run against the orientation in which their phase was defined and contribute the conjugate phase. \textbf{(b)} A redefinition of the phase reference at each node shifts every link phase by $\theta_u+\eta_v$. Each node contributes once with each sign to the alternating sum, leaving $\Phi_s$ unchanged. As these shifts are conventions rather than observable properties, only quantities such as $\Phi_s$ that survive them possess physical meaning.
\textbf{(c)} The same square read as a two-slit interferometer. The two two-step paths $P_1$ and $P_2$ from $u$ to $u'$ combine as $W_s=P_1\overline{P_2}$, whose phase is their difference, so that $\Phi_s$ determines whether the two routes reinforce or cancel.}
\label{fig:schematic}
\end{figure}
 
A square exhibits perfect constructive interference when $\Phi_s=0$ and perfect destructive interference when $\Phi_s=\pi$, and more generally $\cos\Phi_s$ measures the degree of reinforcement. A network is phase coherent when all fluxes vanish and exhibits phase disorder otherwise.
 
\begin{remark}[Relation to structural balance]\label{rem:balance}
A complex-weighted network is structurally balanced when every cycle carries zero flux, equivalently when some choice of phase references renders all weights positive real~\cite{tian2024,lange2015}. Vanishing square fluxes therefore constitutes the signature of balance at the level of the smallest motif. The condition is necessary but not in general sufficient, since squares need not span the cycle space of the network. The phase factor $Q_u$ defined below is accordingly a local and normalised measure of frustration, in the sense of the frustration index of~\cite{lange2015}.
\end{remark}
 
\subsection{Clustering Coefficient for Complex-Weighted Bipartite Networks}\label{sec:clustering}
 
Two requirements determine the form of the measure. Amplitudes must be combined before any modulus is taken, since the modulus of a single square amplitude is a product of moduli and carries no phase information. The normalisation must count the squares that could have formed rather than those that did, so that setting all phases to zero returns a classical clustering coefficient in place of a constant.
 
\begin{definition}[Complex square-based clustering coefficient]\label{def:C}
For a node $u$ with $T(u)\neq\emptyset$,
\begin{equation}
  C_u=\frac{\sum_{s\in S(u)}\Rey W_s}{\sum_{t\in T(u)}|W_t|}  =\frac{\displaystyle\sum_{u'\neq u}\sum_{v_1\neq v_2} \Rey\bigl(w_{uv_1}w_{v_1u'}w_{u'v_2}w_{v_2u}\bigr)} {\displaystyle\sum_{u'\neq u}\sum_{v_1\neq v_2} \bigl|w_{uv_1}w_{v_1u'}w_{u'v_2}\bigr|}
  \label{eq:C}
\end{equation}
and $\bar C$ denotes the average of $C_u$ over all nodes with $T(u)\neq\emptyset$.
\end{definition}
 
The denominator sums the strengths of all three-link paths leaving $u$, and the numerator retains those closing into a square, each weighted by the cosine of its flux. A node whose alternative routes reinforce contributes positively and one whose routes cancel contributes negatively. The coefficient admits a factorisation into two parts of distinct character.
 
\begin{definition}[Topological and phase factors]\label{def:factors}
\begin{equation}
  C^{(0)}_u=\frac{\sum_{s\in S(u)}|W_s|}{\sum_{t\in T(u)}|W_t|},
  \qquad
  Q_u=\frac{\sum_{s\in S(u)}\Rey W_s}{\sum_{s\in S(u)}|W_s|}
     =\frac{\sum_{s\in S(u)}|W_s|\cos\Phi_s}{\sum_{s\in S(u)}|W_s|},
\end{equation}
so that $C_u=C^{(0)}_u\,Q_u$ identically.
\end{definition}
 
The topological factor $C^{(0)}_u$ depends on the moduli and the connectivity alone and coincides with the classical weighted bipartite clustering coefficient. The phase factor $Q_u$ is a strength-weighted average of the plaquette cosines and depends on the fluxes alone. This decomposition permits the ensemble of Section~\ref{sec:model} to be solved.
 
\begin{proposition}[Range, invariance and classical limits]\label{prop:props}
 Assume $\max_{(u,v)\in E}|w_{uv}|\le1$. For every $u$ with $T(u)\neq\emptyset$,
\begin{enumerate}[label=(\roman*)]
\item $Q_u\in[-1,1]$, $C_u^{(0)}\in[0,1]$ and $C_u\in[-1,1]$, negative values indicating that
alternative routes cancel on average. The negative half of this range is not attained in general,
and Proposition~\ref{prop:range} gives the sharp bound.
\item $C_u$, $C_u^{(0)}$ and $Q_u$ are unchanged by any redefinition of the phase references.
\item If all phases vanish then $Q_u=1$ and $C_u=C_u^{(0)}$, the weighted bipartite square
clustering coefficient of Zhang--Horvath type. If in addition all $|w_{uv}|=1$, then
$C_u=m_u/|T(u)|$, the fraction of three-link paths from $u$ that close into a square, which is the
local form of the four-cycle to three-path ratio of~\cite{robins2004}.
\item If the network is structurally balanced then $C_u=C_u^{(0)}$.
\item With uniform strengths, $Q_u=1$ precisely when every square at $u$ has zero flux. The
corresponding condition at the lower end, that every square carry flux $\pi$, is in general
unsatisfiable, since the fluxes on a common pair of nodes are not independent; see
Proposition~\ref{prop:range}.
\item $\sum_{s\in S(u)}W_s$ is real, each square appearing in both directions of traversal and the
two contributing conjugate amplitudes.
\end{enumerate}
\end{proposition}

The upper extreme of Proposition~\ref{prop:props}\,(v) is reached by any balanced network,
whereas the lower extreme is obstructed by the geometry of the flux field. Squares built on the
same pair of nodes satisfy additive relations, since the flux of the square through $v_1$ and
$v_3$ is the sum of the fluxes through $v_1,v_2$ and through $v_2,v_3$, so that three squares on a
common pair cannot all carry flux $\pi$. Destructive interference is self-limiting in this sense,
and the following bound quantifies it.

\begin{proposition}[Attainable range of the phase factor]\label{prop:range}
Fix $u\in\mathcal{U}$ and suppose every edge has strength one. Then
\begin{equation}
Q_u\ \ge\ -\frac{\sum_{u'}q_{uu'}}{\sum_{u'}q_{uu'}(q_{uu'}-1)},
\label{eq:range}
\end{equation}
both sums running over the nodes $u'\neq u$ that share at least two neighbours with $u$, and the
bound is attained by a choice of phases. When $u$ shares neighbours with a single node $u'$ the
bound reads $-1/(q_{uu'}-1)$, so that $Q_u=-1$ is attainable only if no node shares more than two
neighbours with $u$.
\end{proposition}
 
The bound is the statement that the open-path visibility of Section~\ref{sec:openpath} cannot be
negative, as the identity of Appendix~\ref{app:relations} makes explicit, and equality holds precisely
when every pair through which $u$ participates in a square has vanishing visibility. The two
extremes of the coefficient are therefore not symmetric. Perfect reinforcement is available to any
network, while complete cancellation requires a locally sparse arrangement of routes, and a node
of high degree in a densely overlapping neighbourhood cannot register strong net destructive
interference however its phases are chosen.
The same additive relation, $\Phi_{13}=\Phi_{12}+\Phi_{23}$ for the squares on a common pair,
holds identically for every realisation of the phases, so the fluxes at a node are never mutually independent. Theorem \ref{thm:resolution} is derived accordingly, from pairwise covariances.
 
Evaluation of Equation~\eqref{eq:C} does not require squares to be enumerated. Define
\begin{equation}
M=BB^H,\quad N=RR^T,\quad D=R^{\circ2}(R^{\circ2})^T,\quad F=R(R^{\circ2})^T,\quad s_{u'}=\sum_vR_{u'v},
\label{eq:matrices}
\end{equation}
so that $M_{uu'}$ is the total amplitude of all two-step paths between $u$ and $u'$, $N_{uu'}$ the
corresponding sum of moduli, and $s_{u'}$ the strength of node $u'$.
 
\begin{proposition}[Evaluation by sparse matrix products]\label{prop:matrix}
For every $u\in\Uset$,
\begin{align}
  \sum_{s\in S(u)}\Rey W_s &= \sum_{u'\neq u}\bigl(|M_{uu'}|^{2}-D_{uu'}\bigr),
  \label{eq:num}\\
  \sum_{s\in S(u)}|W_s| &= \sum_{u'\neq u}\bigl(N_{uu'}^{2}-D_{uu'}\bigr),
  \label{eq:inc}\\
  \sum_{t\in T(u)}|W_t| &= \sum_{u'\neq u}\bigl(N_{uu'}s_{u'}-F_{uu'}\bigr),
  \label{eq:den}
\end{align}
with the corresponding identities for $\Vset$ obtained on replacing $B$ by $B^{\mathsf H}$.
\end{proposition}
 
By~\eqref{eq:num} the coherent sum over every square at a node equals $\sum_{u'}|M_{uu'}|^{2}$, corrected for the degenerate terms in which the two intermediaries coincide. Under sparse arithmetic all four matrices are formed in $O(nk^2)$ operations for a network of $n$ nodes and degree $k$, so that the node-level profile requires $O(nk^2)$ time and memory, in place of $O(n^2k)$ for explicit enumeration.
 
\subsection{Open Paths and Coherent Distance}\label{sec:openpath}
 
Interference is not confined to closed motifs. Two routes between distinct nodes also combine
as amplitudes, and the distance a coherent signal experiences is not the number of links it crosses.
We therefore introduce two further measures, one for open paths and one for distance.
 
For two nodes $u,u'\in\Uset$ sharing $q_{uu'}\geq1$ neighbours, the total amplitude carried by their two-step paths is $M_{uu'}=\sum_vw_{uv}w_{vu'}$, and the \textit{open-path visibility} is defined as
\begin{equation}
  I^{(2)}_{uu'}=\frac{|M_{uu'}|}{N_{uu'}}
  =\frac{\bigl|\sum_{v}w_{uv}w_{vu'}\bigr|}{\sum_{v}|w_{uv}w_{vu'}|}\in[0,1],
  \label{eq:I2}
\end{equation}
the sums running over the shared neighbours. It equals unity when the routes reinforce perfectly and zero when they cancel completely. A change of phase reference multiplies $M_{uu'}$ by a factor of unit modulus, so that $I^{(2)}_{uu'}$ is invariant in the same sense as the clustering coefficient. The quantity is the bipartite analogue of the visibility of a multi-slit interference pattern, the shared neighbours playing the role of the slits.
 
Distances require more care, since complex numbers cannot be ordered and no canonical shortest path exists. The simplest choice is the hop-count distance, whose average over pairs $L$ counts links and disregards what they carry. It measures combinatorial reachability and, being determined by which links exist, cannot respond to phase. The \emph{coherent distance} between two nodes implies the exponent at which the amplitude reaching one from the other decays, so that a pair whose competing routes cancel is further apart than a pair of the same hop count whose routes reinforce. We realise it through apparent path length of Krawciw \emph{et al.}~\cite{krawciw2024}. The propagator
\begin{equation}
  G=(\mathbb I-\alpha A)^{-1}=\sum_{\ell\geq0}\alpha^{\ell}A^{\ell},
  \qquad \alpha<1/\rho(|A|),
\end{equation}
sums the amplitudes of all walks with a penalty $\alpha^{\ell}$ for length $\ell$, so that $G_{xy}$ is the steady-state amplitude reaching $x$ from a unit source at $y$, the structure that also governs coherent transport in continuous-time quantum walks~\cite{mulken2011}. The apparent distance is
\begin{equation}
  a_{xy}=\frac{\log|G_{xy}|}{\log\alpha}, \qquad
  \bar a=\frac{1}{n_Un_V}\sum_{u\in\Uset,v\in\Vset}a_{uv},
  \label{eq:apparent}
\end{equation}
which returns the hop count when a single walk dominates and the phases vanish, and increases when competing walks cancel. We fix $\alpha=1/(2k)$ throughout, so that comparisons across parameters are made at constant length penalty.
 
\subsection{A Bipartite Watts-Strogatz Ensemble}\label{sec:model}
 
The ensemble is a bipartite Watts--Strogatz model, with rewiring probability $p$ controlling the topology and half-width $\Delta$ controlling the phases. Each node $u\in\Uset$ is first joined to the $k$ nodes $v=(u+j)\bmod n_V$ for $j=0,\dots,k-1$, producing a banded ring lattice with high local closure and no long-range links. Every link is then, independently with probability $p$, detached from its endpoint in $\Vset$ and reattached to a uniformly chosen non-neighbour of the same node in $\Uset$, which preserves both the degrees in $\Uset$ and the bipartite structure. At $p=0$ the network is regular, and at $p=1$ it is a random bipartite network with fixed degrees in $\Uset$. Each link finally receives a phase drawn independently from the uniform distribution on $[-\Delta,\Delta]$, so that $\Delta=0$ yields a phase-coherent network and $\Delta=\pi$ complete phase disorder. Edge weights are set to unity in the main text and drawn uniformly from $(0.5,1)$ in the robustness checks of Appendix~\ref{app:computation}.
 
This ensemble is analytically tractable in the following sense. Rewiring modifies only the network connectivity, and phase disorder modifies only the complex weights assigned to the existing links; the two parameters therefore act on distinct degrees of freedom. Since $C_u=C^{(0)}_uQ_u$, with $C^{(0)}_u$ determined solely by the connectivity and $Q_u$ solely by the link phases, the expectation of the clustering coefficient factorises into a contribution depending on $p$ and a contribution depending on $\Delta$. The phase-dependent contribution, and hence the full effect of phase disorder on the local clustering structure, can then be obtained in closed form, as shown in the following section. Only the connectivity-dependent contribution requires numerical evaluation.

\section{Theoretical Characterisation}\label{sec:theory}
The factorisation $C_u=C^{(0)}_uQ_u$ separates the effects of network topology and phase coherence. $C^{(0)}_u$ depends only on connectivity and retains the structural dependence of the classical Watts--Strogatz model, whereas $Q_u$ depends only on the square fluxes and can be analysed independently. Together with the corresponding result for open paths, the analysis of $Q_u$ determines in closed form the entire effect of phase disorder on local structure in this ensemble, leaving only the topological contribution to be measured.

Throughout, edge weights are set to unity and link phases are drawn independently with characteristic function $\chi(t)=\E[e^{it\varphi}]$. For the uniform distribution on $[-\Delta,\Delta]$,
\[
\chi(t)=\frac{\sin(t\Delta)}{t\Delta},
\qquad\text{so that}\qquad
\chi(1)=\frac{\sin\Delta}{\Delta}\equiv\mathrm{sinc}\Delta
\quad\text{and}\quad
\chi(2)=\frac{\sin (2\Delta)}{2\Delta},
\]
both real. We write
$r=|\chi(1)|^2$, the mean of $e^{i(\varphi-\varphi')}$ for two independent link phases, which
recurs throughout. Proofs are given in Appendix~\ref{app:proofs}.
 
\begin{theorem}[Mean phase factor]\label{thm:meanphasefactor}
For any bipartite network and any node $u$ lying on at least one square,
\begin{equation}
  \E[Q_u]=|\chi(1)|^{4},
  \qquad\text{and for }\varphi\sim\mathcal U[-\Delta,\Delta],\qquad
  \E[Q_u]=\Bigl(\frac{\sin\Delta}{\Delta}\Bigr)^{4}.
  \label{eq:sinc4}
\end{equation}
\end{theorem}
 
The result follows because the four links of a square are always distinct, so that their phases are independent, and the alternating signs in~\eqref{eq:Ws} group them into two conjugate pairs. The expectation of each plaquette cosine is therefore identical, and averaging over the squares at a node requires only linearity, with no assumption that distinct squares are independent. No property of the topology enters, and the expression consequently holds at any degree, system size and rewiring probability. Expanding for small $\Delta$ gives
\begin{equation}
  \E[Q_u]=1-\tfrac{2}{3}\Delta^{2}+\tfrac{1}{5}\Delta^{4}+O(\Delta^{6}),
  \label{eq:expansion}
\end{equation}
so the response to weak phase disorder is quadratic. Local coherence is protected at first order, and a device with modest phase noise therefore loses little local structure, the degradation becoming rapid only as $\Delta$ approaches one radian. At complete disorder $\sinc\pi=0$ and the mean phase factor vanishes.
 
\begin{theorem}[Factorisation of the clustering coefficient]\label{thm:factor}
With unit edge weight, $C^{(0)}_u$ is a deterministic function of the connectivity and independent of the phases. Consequently, for every realisation of the topology,
\begin{equation}
  \E[C_u]=C^{(0)}_u\,|\chi(1)|^{4},
  \qquad
  \E\bigl[\bar C(p,\Delta)\bigr]
     =\E\bigl[\bar C^{(0)}(p)\bigr]\Bigl(\frac{\sin\Delta}{\Delta}\Bigr)^{4}.
  \label{eq:factorisation}
\end{equation}
\end{theorem}
 
Theorem~\ref{thm:factor} makes precise the sense in which topology and coherence act separately, as an identity in expectation rather than an approximate decoupling inferred from data. It has consequences in both directions. Prediction of the two-dimensional behaviour of the clustering coefficient requires only a one-dimensional measurement of the topological factor together with a closed-form curve. Conversely, an observed value of $\bar C$ cannot be attributed to connectivity in the absence of an independent estimate of the coherence, since a sparse coherent network and a dense incoherent one may return the same value.
 
The following result concerns the precision with which the phase factor can be measured.
 
\begin{theorem}[Resolution of the estimator]\label{thm:resolution}
Let $m_u$ denote the number of ordered squares at $u$, and $\tilde m_u=m_u/2$ the number of distinct
squares. Suppose every edge has strength one. Among the pairs of distinct squares at $u$, let $P^{(1)}_u$ be the number that share a
phase difference, namely two squares on the same partner $u'$ with one intermediary in common, or
on different partners with both intermediaries in common, and let $P^{(2)}_u$ be the number that
share exactly one link, necessarily of the form $(u,v)$. Then, for any bipartite network,
\begin{equation}
\Var(Q_u)=\frac{\tilde m_u\,\sigma^2+2P^{(1)}_u\kappa_1+2P^{(2)}_u\kappa_2}{\tilde m_u^{\,2}},
\label{eq:varQ}
\end{equation}

where
\begin{equation}
\begin{aligned}
\sigma^2&=\Var(\cos\Phi)=\tfrac12\bigl(1+|\chi(2)|^4\bigr)-r^4,\\
\kappa_1&=\tfrac12r^2\bigl(1+|\chi(2)|^2\bigr)-r^4,\\
\kappa_2&=\tfrac12r^2\operatorname{Re}\bigl[\chi(2)\,\overline{{\chi(1)}}^{\,2}\bigr]+\tfrac12r^3-r^4
\end{aligned}
\label{eq:kappas}
\end{equation}
are the variance of a plaquette cosine and its covariances for the two kinds of overlap
(Lemma~\ref{lem:cov} in Appendix). The fluxes of two squares sharing no link are independent, so such pairs do not contribute. All
covariances vanish if and only if $\chi(1)=0$, in which case~\eqref{eq:varQ} reduces to
$\Var(Q_u)=2\sigma^2/m_u$. For $\varphi\sim\mathcal{U}[-\Delta,\Delta]$ with $0<\Delta\le\pi$ this happens only
at $\Delta=\pi$, where $\mathrm{rms}(Q_u)=m_u^{-1/2}$ exactly.
\end{theorem}

For symmetric phase distributions, $\chi$ is real, and the two covariances take the transparent
forms $\kappa_1=r^2\Var\bigl(\cos(\varphi-\varphi')\bigr)$ and $\kappa_2=r^3\Var(\cos\varphi)$, both
non-negative, so that every overlap between squares increases the variance. The counts
in~\eqref{eq:varQ} are purely combinatorial and are obtained from the same sparse products as
Proposition~\ref{prop:matrix}, at the same $O(nk^2)$ cost (Appendix~\ref{app:computation}). For the
lattice they are available in closed form.
 
\begin{corollary}[Resolution on the regular lattice]\label{cor:lattice}
On the ring lattice of Section~\ref{sec:model} with $p=0$, degree $k\ge3$ and
$n_U=n_V\ge2k$, every node has $\tilde m=k(k-1)(k-2)/3$,
$P^{(1)}=\tfrac38k(k-1)(k-2)(k-3)$ and $P^{(2)}=\tfrac1{120}k(k-1)(k-2)(27k^2-139k+194)$, so that
\begin{equation}
\Var(Q_u)=\frac{3\bigl[20\sigma^2+45(k-3)\kappa_1+(27k^2-139k+194)\kappa_2\bigr]}{20\,k(k-1)(k-2)}.
\label{eq:varlattice}
\end{equation}
For fixed $\Delta<\pi$, $\Var(Q_u)=81\kappa_2/(20k)+O(k^{-2})$ as $k\to\infty$, whereas at
$\Delta=\pi$, $\Var(Q_u)=1/m_u=3/[2k(k-1)(k-2)]$.
\end{corollary}

Theorem~\ref{thm:resolution} accounts for the dependence between squares. Distinct squares do not carry independent fluxes, by the additive relation noted after Proposition~\ref{prop:range}, but a variance is assembled from pairwise covariances alone, as given in Lemma~\ref{lem:cov}. The result has the following interesting features.
 
The first concerns the case of complete disorder. At $\Delta=\pi$ the link phase that two overlapping squares do not share is uniform on the circle and independent of everything else, and
convolution with the uniform distribution returns the uniform distribution. The two fluxes are then
independent and sharing a link has no statistical effect (Lemma~\ref{lem:cov}). The fluxes at a node are thus pairwise but not mutually independent, and since a variance registers only pairwise dependence, the estimate $2 \sigma^2/m_u$ becomes exact. This is also the regime relevant to detection.
Under the null hypothesis of no coherence the standard deviation of $Q_u$ is exactly $m_u^{-1/2}$,
so that a node lying on $m_u$ squares cannot distinguish genuine coherence from sampling noise below
$m_u^{-1/2}$, and since rewiring destroys squares this threshold increases with $p$.
 
The second concerns the precision with which a non-zero coherence is measured, which is governed by
a different mechanism. For $\Delta<\pi$ the covariances are positive and the dominant contribution
is $P^{(2)}_u$. Every square through a link $(u,v)$ responds to its phase, and since each of the $k$
links at a node is shared by $O(k^2)$ of its $O(k^3)$ squares, the fluctuation of $Q_u$ is set by the
$k$ link phases at $u$ rather than by its $m_u$ squares. Corollary~\ref{cor:lattice} makes this
precise. The standard deviation decays as $k^{-1/2}$, that is as $m_u^{-1/6}$, rather than as
$m_u^{-1/2}$. More generally, to leading order when each link carries many squares,
$\Var(Q_u)\simeq4\kappa_2\sum_vn_{uv}^2/(\sum_vn_{uv})^2$, where $n_{uv}$ is the number of squares
through the link $(u,v)$, an inverse effective number of links. The correction is not a small one. At the degree $k=6$ used in Section~\ref{sec:numerics}, the estimate $2\sigma^2/m_u$ understates the standard deviation of $Q_u$ by a factor $1.67$ at $\Delta=\pi/4$ and $1.42$ at $\Delta=\pi/2$, and the factor grows without bound with $k$. Figure~\ref{fig:resolution} confirms~\eqref{eq:varQ} across the whole range of $\Delta$, on the lattice and node by node on rewired networks.
 
The third concerns the retention of the sign in Definition~\ref{def:C}. An unsigned average such as $|Q_u|$ returns a positive value in a completely incoherent network, and that value grows as the number of squares falls, which may be read as a recovery of coherence when it actually reflects only a smaller sample of fluxes. The signed coefficient avoids this, since its positive and negative contributions cancel and it has mean zero in an incoherent network. For sparse networks, in which $m_u$ is small, the same result determines the number of realisations needed to resolve a given level of coherence, provided~\eqref{eq:varQ} rather than $2\sigma^2/m_u$ is used away from $\Delta=\pi$.
 
The final result extends the argument from closed motifs to open ones.
 
\begin{theorem}[Open-path visibility]\label{thm:I2}
If $u$ and $u'$ share exactly $q\geq1$ neighbours and edge weights are unity, then
with $r=|\chi(1)|^{2}$,
\begin{equation}
  \E\bigl[(I^{(2)}_{uu'})^{2}\bigr]=r^{2}+\frac{1-r^{2}}{q},
  \qquad
  r^{2}=\Bigl(\frac{\sin\Delta}{\Delta}\Bigr)^{4}.
  \label{eq:I2law}
\end{equation}
\end{theorem}
 
The expression interpolates between perfect visibility at full coherence, for any number of shared neighbours, and a limiting value of $q^{-1/2}$ at complete disorder. This value is the shot noise arising from the addition of $q$ random phasors, so that the residual visibility observed at large $\Delta$ reflects the finite number of alternative routes and not any surviving structure. Combined with Theorem~\ref{thm:resolution}, it establishes that at complete disorder the closed and open measures are limited by the same form of combinatorial sampling noise, and that in both cases the limit is set by the number of alternative routes rather than by the system size. Away from complete disorder the two measures behave differently. The visibility of a pair involves $2q$ distinct links and no shared ones, so that Equation~\eqref{eq:I2law} holds without correction, whereas the squares at a node overlap on its links.
 
The topological factor $\bar C^{(0)}(p)$ admits no simple closed form, and global structure permits no comparable factorisation. These questions are addressed numerically in the following section.

\section{Numerical Results}\label{sec:numerics}
The simulations help determine the topological factor, for which no closed form is available, and confirm that combining it with the exact phase factor reproduces the clustering coefficient across the parameter plane. We then compute the coherent distance of Section~\ref{sec:openpath}, since the factorisation does not apply to global structure and the distance cannot be derived in closed form. Finally, we test the exact results outside the idealisations under which they were proved, in particular for heterogeneous strengths.

Unless stated otherwise we take $n_U=n_V=200$ nodes in each class, degree $k=6$, length penalty $\alpha=1/(2k)$, twenty-one values of the rewiring probability $p$ spaced logarithmically in $[10^{-3},1]$, and phase disorder $\Delta\in\{0,\pi/4,\pi/2,3\pi/4,\pi\}$, averaging over forty independent realisations at each parameter pair. Error bars denote standard errors of the mean and are generally smaller than the plotting symbols. Validation, robustness checks and computational cost are reported in Appendix~\ref{app:computation}.
 
\subsection{Topology and Coherence in the Clustering Coefficient}
 
\begin{figure}[t]
\centering
\includegraphics[width=\textwidth]{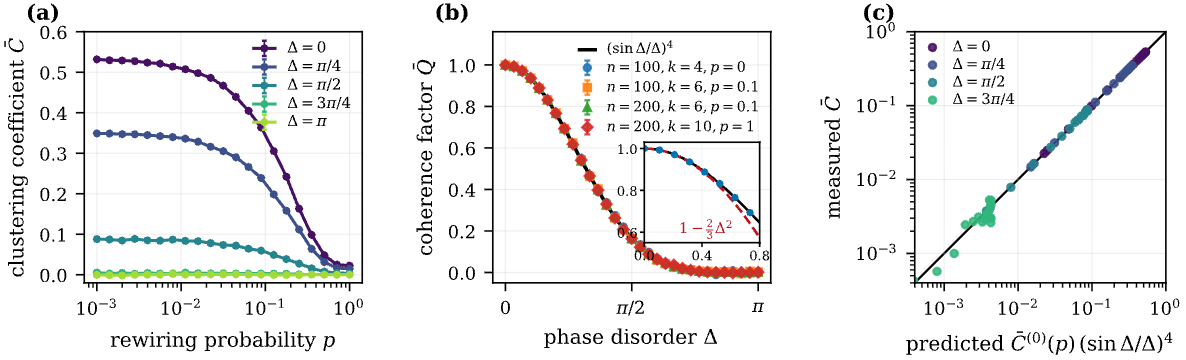}
\caption{The clustering coefficient is the product of a topological factor and a phase factor. Here $p$ is the probability that a link is rewired, $\Delta$ the half width of the distribution from which link phases are drawn, $\bar C$ the average clustering coefficient of Equation~\eqref{eq:C}, $\bar C^{(0)}$
its topological factor and $\bar Q$ its phase factor. \textbf{(a)} Clustering is suppressed by two independent mechanisms. Rewiring destroys squares, driving the coherent curve from $0.53$ to $0.02$, while phase disorder lowers every curve without altering its shape. At $\Delta=\pi/2$ the coefficient is six times smaller than at $\Delta=0$ although the two ensembles have statistically identical connectivity, so that a low value of $\bar C$ alone does not indicate a sparse or randomised topology.
\textbf{(b)} The phase factor measured for four combinations of system size, degree and rewiring probability. The four data sets are indistinguishable and follow the prediction $(\sin\Delta/\Delta)^4$ of Theorem~\ref{thm:meanphasefactor}, which contains no fitted parameter, confirming that the phase factor is independent of the topology on which it is evaluated. Inset: at small $\Delta$ the curve is flat at first order, with quadratic onset $1-\frac{2}{3}\Delta^2$ shown dashed. 
\textbf{(c)} Test of Theorem~\ref{thm:factor}. The measured $\bar C$ is plotted against the prediction $\bar C^{(0)}(p)(\sin\Delta/\Delta)^4$ for all $21\times4$ parameter pairs, together with the identity line. The case $\Delta=\pi$ is omitted, the prediction then being identically zero.}
\label{fig:clustering}
\end{figure}
 
Figure~\ref{fig:clustering} presents the clustering coefficient across the parameter plane together with its two factors. Panel (a) gives the phenomenology. Clustering decreases with rewiring in the manner familiar from the Watts--Strogatz model, and phase disorder lowers the entire family of curves without altering their shape. The suppression is substantial well before the phases become fully random: at $\Delta=\pi/2$ the phase factor is already $0.164$, so that more than four fifths of the local closure supported by the
connectivity has been removed by interference alone.
 
Panels (b) and (c) establish that the two effects separate exactly. The measured phase factor collapses onto $(\sin\Delta/\Delta)^4$ for system sizes of $100$ and $200$ nodes per class, degrees $4$, $6$ and $10$, and rewiring probabilities between zero and one, as required by Theorem~\ref{thm:meanphasefactor}. Plotting the measured clustering coefficient against the product of the measured topological factor and the predicted phase factor places every parameter pair on the identity line across three orders of magnitude, in accordance with Theorem~\ref{thm:factor}. Only the topological factor required simulation, the phase dependence at every value of $p$ following without further computation.
 
The inset of panel (b) carries a practical implication. Since the phase factor is flat at the origin, a device with phase noise of $\Delta\lesssim0.3$ retains more than ninety three per cent of its coherent clustering. Local interference structure is thus robust to weak disorder and degrades rapidly thereafter.
 
\subsection{Coherence and Global Structure}
 
\begin{figure}[t]
\centering
\includegraphics[width=\textwidth]{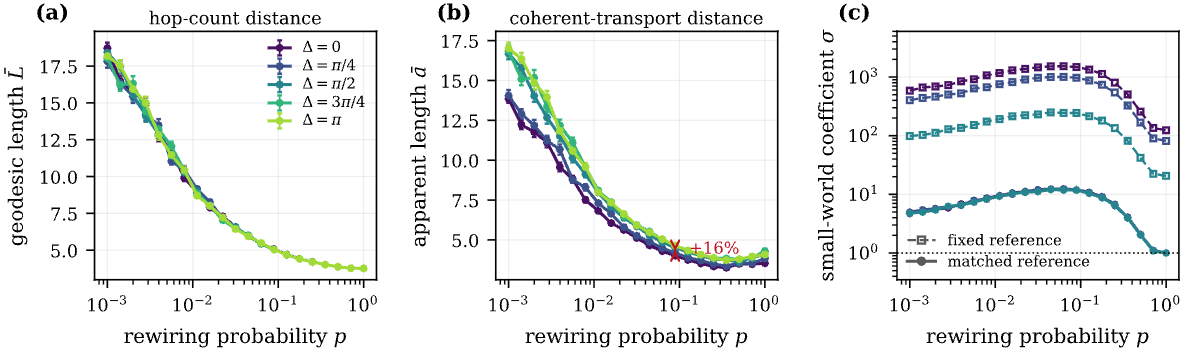}
\caption{Phase disorder leaves the number of hops between nodes unchanged while increasing the distance experienced by a coherent signal. Here $\bar L$ is the mean hop-count distance, $\bar a$ the mean apparent distance of Equation~\eqref{eq:apparent} obtained from the propagator $(\mathbb I-\alpha A)^{-1}$, and $\sigma$ the small-world coefficient of Humphries and Gurney~\cite{humphries2008}.
\textbf{(a)} Hop-count distance. All five levels of phase disorder collapse onto a single curve, this distance being determined by which links exist. The small-world decrease occurs near $p\approx10^{-2}$.
\textbf{(b)} Apparent distance, which sums over all walks rather than the shortest. The decrease with rewiring persists, but the curves now separate with phase disorder. At $p\approx0.09$, marked by the arrow, cancellation between competing walks increases the effective distance by sixteen per cent, from $4.01$ to $4.64$, while the hop count on the same networks changes by less than one per cent. Coherence therefore acts on global structure, but only on a distance measure that accounts for the interference of competing walks.
\textbf{(c)} The small-world coefficient for $\Delta=0$, $\pi/4$ and $\pi/2$, with colours as in panel (a). Solid curves employ a reference network carrying the same phase disorder as the network under study, and the three coincide because the phase factor cancels between the two, as predicted by Equation~\eqref{eq:sigma-cancel}. Dashed curves use one reference network, generated at $\Delta=\pi$ and reused at every value of $\Delta$. Larger $\Delta$ is omitted, as both ratios approach the indeterminate form $0/0$.} \label{fig:transport}
\end{figure}
 
The Watts--Strogatz description has two components: local structure and distance. Figure~\ref{fig:transport}(a) shows that the mean hop-count distance decreases from $18.7$ to $3.8$ under rewiring but is unaffected by phase disorder, as a hop count depends only on which links exist. Coherence therefore leaves hop-count distance unchanged, and its effect on global structure must be reflected in another notion of distance.
 
Figure~\ref{fig:transport}(b) shows where this effect appears. The apparent distance retains the decrease with rewiring but increases with phase disorder. On networks of fixed connectivity at $p\approx0.09$, the effective distance increases by $15.7$ per cent between full coherence and complete disorder, while the hop count on the same networks changes by less than one per cent. Flux disorder therefore does not disconnect the network; instead, it introduces cancellation between alternative routes. The network remains equally navigable in a combinatorial sense, while becoming measurably harder to traverse coherently. This provides a static structural counterpart to the suppression of transfer by destructive interference familiar from coherent transport~\cite{mulken2011,mohseni2008}. 
 
\subsection{The Small-World Coefficient and its Reference Network}\label{sec:sw}
 
These two quantities are conventionally combined into the small-world coefficient of Humphries and Gurney~\cite{humphries2008},
\[
  \sigma=\frac{\bar C/\bar C_{\mathrm{rand}}}{\bar L/\bar L_{\mathrm{rand}}},
\]
in which $\bar C_{\mathrm{rand}}$ and $\bar L_{\mathrm{rand}}$ denote the clustering and the distance of a randomised reference network. Under complex weights the choice of reference proves consequential in a manner with no real-weighted analogue, and here the analysis of Section~\ref{sec:theory} highlights this distinction.
 
Let the reference be a randomised topology carrying the same phase disorder, which is the faithful generalisation of the Watts--Strogatz procedure since only the connectivity is randomised. By Theorem~\ref{thm:meanphasefactor} the phase factor is independent of topology, so that it appears identically in numerator and denominator and cancels,
\begin{equation}
  \sigma(p,\Delta)
  =\frac{\bar C^{(0)}(p)\,\bar Q(\Delta)}{\bar C^{(0)}(1)\,\bar Q(\Delta)}
   \Big/\frac{\bar L(p)}{\bar L(1)}
  =\frac{\bar C^{(0)}(p)/\bar C^{(0)}(1)}{\bar L(p)/\bar L(1)}.
  \label{eq:sigma-cancel}
\end{equation}
The small-world coefficient of this ensemble is therefore exactly independent of phase disorder and equal to its classical real-weighted value. The solid curves of Figure~\ref{fig:transport}(c) confirm this, with peak values of $12.37$, $12.37$ and $12.08$ for $\Delta=0$, $\pi/4$ and $\pi/2$, agreeing to within statistical error.
 
If instead a single reference network is generated with random phases and reused at every value of $\Delta$, the cancellation is lost and the reference clustering tends to zero in expectation. In the present ensemble $\bar C(p{=}1,\Delta{=}\pi)=1.8\times10^{-4}$, and $\sigma$ inflates to $1.5\times10^{3}$ at full coherence, decreasing as the phase disorder increases. The small-world coefficient \(\sigma\) is known to be unduly sensitive to small reference clustering in real-weighted networks~\cite{telesford2011}; complex weights make this problem qualitatively more severe, since the denominator vanishes in expectation rather than merely becoming small.
 
Any phase dependence of $\sigma$ in a complex-weighted network should first be verified against a matched reference, since a mismatched reference can produce such dependence artificially. Furthermore, small-worldness as measured through clustering does not register coherence in this model, because clustering reflects connectivity rather than the phase structure of the weights. The actual phase dependence of global structure is instead captured by the quantity shown in Figure~\ref{fig:transport}(b).
 
\subsection{Interference between Open Paths}
 
\begin{figure}[t]
\centering
\includegraphics[width=\textwidth]{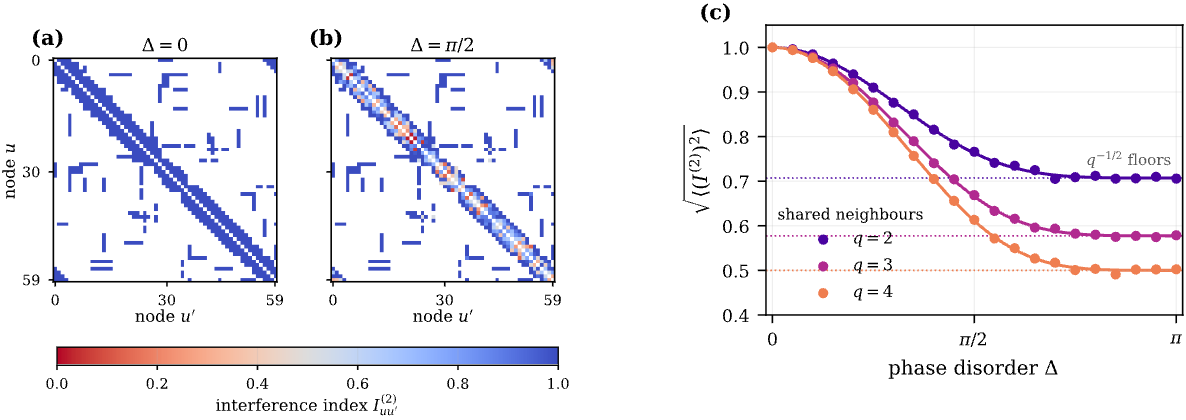}
\caption{Topological disorder determines which pairs of nodes interfere, and phase disorder determines how strongly pairs interfere. The quantity plotted is the open-path visibility $I^{(2)}_{uu'}$ of Equation~\eqref{eq:I2}, the modulus of the summed amplitude of all two-step paths between two nodes of $\Uset$ divided by the sum of the individual moduli, and $q$ is the number of neighbours the two nodes share. \textbf{(a,b)} The visibility for every pair of nodes in a network with $60$ nodes per class, degree $5$ and rewiring probability $0.1$. Blue indicates reinforcement and red indicates cancellation. Pairs sharing no neighbour, for which the visibility is undefined, are left blank. The connectivity is identical in the two panels and only the phase disorder differs. At $\Delta=0$ every pair sharing a neighbour has visibility one, producing a uniform band along the diagonal inherited from the ring lattice together with isolated off-band pairs created by rewiring. At $\Delta=\pi/2$ the same pairs are present but their visibility is distributed across the full range. \textbf{(c)} Root-mean-square visibility against phase disorder, resolved by the number of shared neighbours. Symbols denote simulations and solid lines the prediction of Theorem~\ref{thm:I2}, which contains no fitted parameter. Dotted lines mark the asymptotic values $q^{-1/2}$ approached by the curves. Under complete decoherence, the residual signal arises in the sum of $q$ random phasors, so pairs joined by more alternative routes lose visibility faster and approach lower limiting values.}
\label{fig:interference}
\end{figure}
 
The clustering coefficient averages interference over the squares incident on a node, and Theorem~\ref{thm:I2} gives the corresponding average for pairs of nodes. Neither quantity reveals how interference is distributed across the network; this spatial structure is instead made apparent in Figure~\ref{fig:interference}(a,b). Fixing the connectivity and varying only the phase disorder leaves the set of interfering pairs unchanged and alters only the strength of the interference between each pair. The banded structure inherited from the ring lattice and the isolated pairs created by rewiring are identical in both panels. Changing the connectivity therefore changes which pairs can interfere, while changing the phase disorder alters the strength of the interference between them.
 
Panel (c) quantifies this distinction and confirms Theorem~\ref{thm:I2} without any fitted parameter, across the full range of phase disorder and for each number of shared neighbours. The approach to the limiting behaviour $q^{-1/2}$ is worth examining. Pairs joined by more alternative routes lose visibility faster and terminate lower. Thus redundancy of routes, which can only help in a real-weighted network, may become problematic once interference allows routes to cancel. This is the structural counterpart of the observation in coherent transport that additional pathways do not necessarily improve transfer~\cite{mohseni2008}. Redundancy has the opposite effect on the closed measure, since the same non-negativity of the visibility that gives rise to these limiting values also bounds the clustering coefficient from below through Proposition~\ref{prop:range}. Many alternative routes therefore reduce the visibility of a pair while constraining how negative the coefficient at a node can become.
 
Bringing these results together shows that the topological factor $\bar C^{(0)}$ depends on $p$ alone, decreasing from $0.53$ to $0.02$ across the rewiring range; the phase factor $\bar Q$ depends on $\Delta$ alone and follows $(\sin\Delta/\Delta)^4$; their product is the clustering coefficient. The hop-count distance $\bar L$ depends on $p$ alone by construction, decreasing from $18.7$ to $3.8$, whereas the coherent distance $\bar a$ responds to both, increasing by $15.7$ per cent across the range of phase disorder. The small-world coefficient is exactly independent of phase against a matched reference and spans two orders of magnitude against a fixed one, and the open-path visibility follows $\sqrt{r^2+(1-r^2)/q}$ with $r=(\sin\Delta/\Delta)^2$.

\section{Discussion}\label{sec:conclusion}

We have introduced a clustering coefficient for networks that are both bipartite and complex weighted, a combination that arises in quantum and photonic architectures but is not covered by existing measures. Bipartite clustering coefficients are defined on four-node squares under the assumption of real weights, whereas interferometric measures that account for phases are built on triangles, which cannot occur in a bipartite network.

We find that the form of the coefficient is constrained by certain requirements. First, the phase assigned to an individual link is a matter of convention, so any structural measure of phase must be invariant under a redefinition of the phase reference at each node. In a bipartite network, the smallest quantity with this invariance is the phase accumulated around a four-node square. The square is therefore the elementary carrier of gauge-invariant phase information, rather than merely the shortest possible cycle. The second requirement concerns the, interference arising from the sum of amplitudes: the modulus of an individual square amplitude is a product of link strengths and contains no phase information. A phase-sensitive coefficient must therefore combine amplitudes before taking their modulus. The resulting coefficient reduces to the classical bipartite coefficient when the phases vanish, lies in $[-1,1]$, and becomes negative when alternative routes to a node interfere destructively. The coherent sum over all squares incident to a node can be computed from a single sparse matrix product, with a cost that is linear in the number of nodes.

The coefficient further separates into a factor determined by connectivity and a factor determined by the phases, making the phase dependence analytically tractable. Applied to a bipartite Watts--Strogatz ensemble, in which links are rewired with probability $p$ and phases are drawn uniformly from $[-\Delta,\Delta]$, the mean phase factor is $(\sin\Delta/\Delta)^4$, independently of node degree and topology. We obtain three further results based on this factorisation. The clustering coefficient is the product of the topological factor and this phase-dependent curve; the visibility of an open path between two nodes sharing $q$ neighbours decreases from unity to the limiting value $q^{-1/2}$; and the variance of the phase factor equals $1/m_u$, where $m_u$ is the number of squares at a node, only at complete phase disorder. Away from this limit, the variance is dominated by pairs of squares that share a link, while the standard deviation decays as $k^{-1/2}$ with degree $k$.

We determine the topological factor and the global structure numerically, since neither admits
a closed form, completing the description of the clustering coefficient across the parameter plane. The topological factor falls from $0.53$ to $0.02$ across the rewiring range, while the mean hop-count distance falls from $18.7$ to $3.8$. Phase disorder leaves the hop-count distance unchanged, since it depends only on the pattern of links, but increases the apparent distance by sixteen per cent on networks with identical connectivity. Flux disorder therefore does not disconnect the network; it introduces destructive interference between alternative routes, leaving the combinatorial structure unchanged while altering coherent transport through it.

These results extend beyond the specific measures and have broader consequences of incorporating phase into network structure. Invariance under local redefinition of phase references is necessary for a structural measure of phase, since phase origins may be fixed by calibration rather than by the underlying hardware. Likewise, combinatorial and coherent distance capture different aspects of connectivity: the coupling map of a photonic mesh may remain unchanged while the effective propagation between nodes becomes less coherent. A corresponding distinction arises at the global level. The phase factor cancels from the small-world coefficient when a network is compared with a matched reference, so any residual phase dependence reflects the choice of reference, not an intrinsic phase dependence of the network. As the reference clustering has mean zero at full phase disorder, this normalisation can produce distortions spanning several orders of magnitude.

The scope of our results is set by several assumptions. Phases are drawn independently for each link, are static, and are independent of the link strengths. Real devices have correlated phases, fixed by geometry, detuning or a shared local oscillator, and correlations of this kind would introduce
further dependence between the fluxes of neighbouring squares. Theorem~\ref{thm:resolution} already
accounts exactly for the dependence created by shared links, and its derivation extends directly to
correlated phases, but the covariances $\kappa_1$ and $\kappa_2$ would then have to be recomputed, and pairs of squares sharing no link would also contribute. The mean phase factor of
Theorem~\ref{thm:meanphasefactor} is unaffected, provided the four links within a single square remain independent. The treatment is moreover static, the propagator describing steady-state amplitude transfer rather than dynamics, and it describes single-particle or wave-like amplitudes
rather than entangled multipartite states.

There are several directions for future exploration. Since the observable content of the phases is the flux, not the individual phases, a more natural ensemble places the disorder directly on the squares, which would connect the framework to the random-flux models of localisation theory~\cite{altland1999,furusaki1999} and allow spatial correlations in the flux to be studied. Squares are the smallest invariant but not the only one, since cycles of length six and eight carry independent flux information, so a hierarchy of coefficients indexed by cycle length would give a fuller account of frustration, with the cycle-space argument of Remark~\ref{rem:balance} fixing which of the higher-order coefficients carry information not already contained in the squares. Finally, the phase factor is measurable on existing hardware: circuit quantum electrodynamics devices are characterised by both edge strengths and phases, so the coefficient can be computed from the measured strengths and phases of a device, and the prediction that interference in such a device is controlled by the fluxes through its squares, not by the phases of individual couplings, admits experimental test.

\section*{Acknowledgements}
RL acknowledges support from the EPSRC grants EP/V013068/1, EP/V03474X/1, and EP/Y028872/1.
YT acknowledges support from ELBE postdoctoral fellowship. 

\bibliographystyle{unsrt}
\bibliography{references}
 
\appendix
\section{Proofs}\label{app:proofs}

\begin{proof}[Proof of Proposition~\ref{prop:gauge}]
Under $\varphi_{uv}\mapsto\varphi_{uv}+\theta_u+\eta_v$ the flux~\eqref{eq:Ws}
becomes
\[
  (\varphi_{uv_1}+\theta_u+\eta_{v_1})-(\varphi_{u'v_1}+\theta_{u'}+\eta_{v_1}) +(\varphi_{u'v_2}+\theta_{u'}+\eta_{v_2})-(\varphi_{uv_2}+\theta_u+\eta_{v_2})
 =\Phi_s,
\]
since each of $\theta_u$, $\theta_{u'}$, $\eta_{v_1}$ and $\eta_{v_2}$ occurs once with each sign. Moduli are unaffected because $|e^{i\theta}|=1$.
\end{proof}
 
\begin{proof}[Proof of Proposition~\ref{prop:props}]
(i) Each square contributes $|W_s|\cos\Phi_s$ to the numerator of $Q_u$ and $|W_s|$ to its denominator, whence $|Q_u|\leq1$. Dropping the closing link embeds $S(u)$ into $T(u)$, and $|W_s|=|W_t|\,|w_{v_2u}|\leq|W_t|$ when $\max|w|\leq1$, so that summation gives $C^{(0)}_u\in[0,1]$ and hence $C_u=C^{(0)}_uQ_u\in[-1,1]$.
(ii) Immediate from Proposition~\ref{prop:gauge}, all three quantities being built from the $W_s$ and the $|W_t|$.
(iii) If all phases vanish then every $\Phi_s=0$ and $Q_u=1$. If in addition $|w|\equiv1$ then $|W_s|=|W_t|=1$ and $C_u=m_u/|T(u)|$.
(iv) In a balanced network some choice of phase references renders all weights positive real. By (ii) the coefficients are unchanged by that choice, and by (iii) they satisfy $C_u=C^{(0)}_u$ in it.
(v) With uniform strengths $Q_u$ is the average of $\cos\Phi_s$ over the squares at $u$, and its extreme values require every term to attain the same extreme.
(vi) The reversed tuple $(v_2,u',v_1)$ also lies in $S(u)$ and has amplitude $B_{uv_2}\overline{B_{u'v_2}}B_{u'v_1}\overline{B_{uv_1}}=\overline{W_s}$, so that squares occur in conjugate pairs.
\end{proof}

\begin{proof}[Proof of Proposition~\ref{prop:range}]
Fix a node $u'$ sharing the neighbours $v_1,\dots,v_q$ with $u$, where $q=q_{uu'}\geq2$, and set $\psi_j=\varphi_{uv_j}-\varphi_{u'v_j}$. By \eqref{eq:Ws} the square through $v_j$ and $v_l$ has flux $\psi_j-\psi_l$. With unit strengths the ordered squares running through $u'$ number $q(q-1)$, and they contribute
\[
  \sum_{j\neq l}\cos(\psi_j-\psi_l)
  =\Bigl|\sum_{j=1}^{q}e^{i\psi_j}\Bigr|^{2}-q\;\geq\;-q ,
\]
with equality exactly when the $q$ unit phasors sum to zero, which is possible for every $q\geq2$, for instance by spacing the $\psi_j$ equally around the circle. Summing over the nodes $u'$ and dividing by $m_u=\sum_{u'}q_{uu'}(q_{uu'}-1)$ gives~\eqref{eq:range}. The phases $\varphi_{u'v_j}$ belong to distinct links for distinct $u'$ and may be chosen independently, so the minimum is realised simultaneously for every $u'$, and the bound is attained. Nodes with $q_{uu'}\leq1$ contribute no squares and are absent from both sums. The single-partner case follows on retaining one term.
\end{proof}

\begin{proof}[Proof of Proposition~\ref{prop:matrix}]
Write $W_s=B_{uv_1}\overline{B_{u'v_1}}B_{u'v_2}\overline{B_{uv_2}}$ and sum over all $v_1$ and $v_2$, including the diagonal terms, to obtain 
\[
  \sum_{v_1,v_2}W_s
  =\Bigl(\sum_{v_1}B_{uv_1}\overline{B_{u'v_1}}\Bigr) \Bigl(\sum_{v_2}B_{u'v_2}\overline{B_{uv_2}}\Bigr)
  =M_{uu'}M_{u'u}=|M_{uu'}|^{2},
\]
where the last step uses that $M$ is Hermitian. The diagonal terms with $v_1=v_2$ contribute $\sum_v|B_{uv}|^{2}|B_{u'v}|^{2}$, which is $D_{uu'}$. Subtracting them and summing over $u'\neq u$ gives~\eqref{eq:num}, which is automatically real. Identity~\eqref{eq:inc} is the same computation with $B$ replaced by $R$. For~\eqref{eq:den} we have
$\sum_{v_1,v_2}R_{uv_1}R_{u'v_1}R_{u'v_2}=N_{uu'}s_{u'}$, from which the diagonal contribution $\sum_vR_{uv}R_{u'v}^{2}=F_{uu'}$ is subtracted.
\end{proof}
 
\begin{proof}[Proof of Theorem~\ref{thm:meanphasefactor}]
With unit strengths $|W_s|=1$ for every square, so that
\[
  Q_u=\frac{1}{m_u}\sum_{s\in S(u)}\cos\Phi_s ,
\]
where $m_u$ is determined by the topology. Fix a square with tuple $(v_1,u',v_2)$. Since $u\neq u'$ and $v_1\neq v_2$, its four links are pairwise distinct, so their phases are independent, whence
\[
  \E\bigl[e^{i\Phi_s}\bigr] =\chi(1)\overline{\chi(1)}\,\chi(1)\overline{\chi(1)}=|\chi(1)|^{4},
\]
which is real, so that $\E[\cos\Phi_s]=|\chi(1)|^{4}$. Averaging over the squares at $u$ uses only linearity of expectation, without any assumption of independence between squares. For $\varphi\sim\mathcal U[-\Delta,\Delta]$,
$\chi(1)=(2\Delta)^{-1}\int_{-\Delta}^{\Delta}e^{i\varphi}\,d\varphi
=\sin\Delta/\Delta$.
\end{proof}
 
\begin{proof}[Proof of Theorem~\ref{thm:factor}]
With unit strengths $C^{(0)}_u=m_u/|T(u)|$ depends only on which links exist and hence not on the phases, so that $\E[C_u]=C^{(0)}_u\E[Q_u]=C^{(0)}_u|\chi(1)|^{4}$ by Theorem~\ref{thm:meanphasefactor}. Averaging over nodes and then over topologies, which are generated independently of the phases, gives the global statement.
\end{proof}

The proof of Theorem~\ref{thm:resolution} rests on the following lemma, which classifies the ways in which two squares at a node can overlap.

\begin{lemma}[Pairwise covariance of plaquette cosines]\label{lem:cov}
Let $s=(u',\{v_1,v_2\})$ and $s'=(u'',\{w_1,w_2\})$ be distinct squares at $u$, with unit strengths
and independent link phases of characteristic function $\chi$. Exactly one of the following holds.
\begin{enumerate}[label=(\alph*)]
\item $s$ and $s'$ share a phase difference: either $u'=u''$ with one intermediary in common, or
$u'\neq u''$ with $\{v_1,v_2\}=\{w_1,w_2\}$. Then $\mathrm{Cov}(\cos\Phi_s,\cos\Phi_{s'})=\kappa_1$.
\item $u'\neq u''$ with exactly one intermediary $v$ in common, so that the squares share the single link $(u,v)$. Then $\mathrm{Cov}(\cos\Phi_s,\cos\Phi_{s'})=\kappa_2$.
\item $s$ and $s'$ share no link, and $\Phi_s$ and $\Phi_{s'}$ are independent.
\end{enumerate}

Here $\kappa_1,\kappa_2$ are given by~\eqref{eq:kappas}. If the phases are uniform on the circle, $\Phi_s$ and $\Phi_{s'}$ are independent in every case.
\end{lemma}

\begin{proof}
\emph{Classification.} The links of $s$ are $(u,v_1),(u,v_2),(u',v_1),(u',v_2)$, and similarly for
$s'$. If $u'=u''$ the squares cannot share both intermediaries without coinciding, so they share
exactly one intermediary, $v_1$ say, or none. In the first case the common links are $(u,v_1)$ and
$(u',v_1)$, which enter both fluxes through the difference $\psi_1=\varphi_{uv_1}-\varphi_{u'v_1}$. If
$u'\neq u''$ the links at $u'$ and $u''$ are distinct, and the common links are the $(u,v)$ with $v$
in both pairs of intermediaries. Two common intermediaries give the shared difference
$\varphi_{uv_1}-\varphi_{uv_2}$, one gives a single shared link, and none gives disjoint link sets,
whose phases are independent. This establishes the three cases and (c).
 
\emph{Case (a).} Since $\cos$ is even, orient both fluxes so that the shared difference enters with
a positive sign. Then $\Phi_s=a-b$ and $\Phi_{s'}=a-c$, where $a$, $b$ and $c$ are independent and
each is a difference of two independent link phases, so that $\E[e^{ia}]=|\chi(1)|^2=r$ and
$\E[e^{2ia}]=|\chi(2)|^2$, both real. Using
$\cos\Phi_s\cos\Phi_{s'}=\tfrac12\cos(2a-b-c)+\tfrac12\cos(b-c)$,
\[
\E[\cos\Phi_s\cos\Phi_{s'}]=\tfrac12|\chi(2)|^2r^2+\tfrac12r^2,
\]
and subtracting $\E[\cos\Phi_s]\E[\cos\Phi_{s'}]=r^4$ (Theorem~\ref{thm:meanphasefactor}) gives $\kappa_1$.
 
\emph{Case (b).} Orient both fluxes so that $x=\varphi_{uv}$ enters positively. Then $\Phi_s=x+y$ and
$\Phi_{s'}=x+y'$, where $y$ and $y'$ are independent of each other and of $x$, and each is a sum of
three link phases with signs $(-,+,-)$, so that $\E[e^{iy}]=\overline{{\chi(1)}}\,r$. Hence
$\mathbb{E}[e^{i(\Phi_s+\Phi_{s'})}]=\chi(2)\,\overline{{\chi(1)}}^{\,2}r^2$ and
$\E[e^{i(\Phi_s-\Phi_{s'})}]=|\E[e^{iy}]|^2=r^3$. Using
$\cos\Phi_s\cos\Phi_{s'}=\tfrac12\operatorname{Re}\bigl(e^{i(\Phi_s+\Phi_{s'})}+e^{i(\Phi_s-\Phi_{s'})}\bigr)$
and subtracting $r^4$ gives $\kappa_2$.
 
\emph{Uniform phases.} In case (a) the difference $b$ of two uniform phases is uniform and
independent of $(a,c)$. Convolution with the uniform distribution on the circle returns the uniform
distribution, so $a-b$ is uniform and independent of $(a,c)$, hence of $\Phi_{s'}=a-c$. In case (b)
the phase $\varphi_{u'v_2}$ appears in $\Phi_s$ alone, and the same argument shows that $\Phi_s$ is
uniform and independent of $(x,y')$, hence of $\Phi_{s'}$.
\end{proof}

The independence at uniform phases is pairwise only. Three squares on a common pair satisfy
$\Phi_{13}=\Phi_{12}+\Phi_{23}$ identically, so any two of them determine the third.
 
\begin{proof}[Proof of Theorem~\ref{thm:resolution}]
By Proposition~\ref{prop:props}\,(vi) the $m_u$ ordered squares form $\tilde m_u$ conjugate pairs,
and $\cos\Phi$ takes the same value on both members of a pair, so that
$Q_u=\tilde m_u^{-1}\sum_s\cos\Phi_s$, the sum running over distinct squares. Then
\[
\Var(Q_u)=\frac1{\tilde m_u^{\,2}}\Bigl(\sum_s\Var(\cos\Phi_s)+\sum_{s\neq s'}\mathrm{Cov}(\cos\Phi_s,\cos\Phi_{s'})\Bigr),
\]
the second sum running over ordered pairs, and Lemma~\ref{lem:cov} gives~\eqref{eq:varQ}. For the
variance, $\E[\cos^2\Phi]=\tfrac12(1+\E[\cos2\Phi])$, and the argument of Theorem~\ref{thm:meanphasefactor} at
frequency two gives $\E[\cos2\Phi]=|\chi(2)|^4$, so that $\sigma^2=\tfrac12(1+|\chi(2)|^4)-r^4$.
 
\emph{Vanishing of the covariances.} Both $\kappa_1$ and $\kappa_2$ carry a factor $r^2$ and vanish
when $\chi(1)=0$. Conversely $\kappa_1=r^2\Var(\cos a)$, where $a=\varphi-\varphi'$ is the difference of two
independent link phases. Zero lies in the support of $a$, and if $\varphi$ is not a point mass the
support also contains some $x-y\neq0$ with $x$ and $y$ in the support of $\varphi$; since
$x,y\in[-\pi,\pi)$ we have $0<|x-y|<2\pi$ and hence $\cos(x-y)<1$. Every neighbourhood of either
point carries positive probability, so $\cos a$ is not almost surely constant and $\kappa_1>0$
whenever $r>0$. For $\varphi\sim \mathcal{U}[-\Delta,\Delta]$, $\chi(1)=\sin\Delta/\Delta$ vanishes on $(0,\pi]$ only at
$\Delta=\pi$, where also $\chi(2)=0$, so that $\sigma^2=\tfrac12$ and $\Var(Q_u)=1/m_u$.
 
\emph{Counts.} Write $q_{u'}=q_{uu'}$, let $c_{v_1v_2}$ be the number of nodes other than $u$
adjacent to both $v_1$ and $v_2$, and let $n_{uv}=\sum_{u'\in N(v)\setminus\{u\}}(q_{u'}-1)$ be the
number of distinct squares at $u$ through the link $(u,v)$. Pairs of type (a) on a common partner
number $P_A=\sum_{u'}q_{u'}(q_{u'}-1)(q_{u'}-2)/2$, and those on a common pair of intermediaries
number $P_B=\sum_{\{v_1,v_2\}\subseteq N(u)}\binom{c_{v_1v_2}}{2}$, so that $P^{(1)}_u=P_A+P_B$. Every
pair of squares through a common link $(u,v)$ is counted in $\sum_v\binom{n_{uv}}{2}$, pairs of the
first kind of type (a) once, pairs of the second kind twice (they share two such links), and pairs
of type (b) once, whence $P^{(2)}_u=\sum_{v\in N(u)}\binom{n_{uv}}{2}-P_A-2P_B$.
\end{proof}
 
\begin{proof}[Proof of Corollary~\ref{cor:lattice}]
On the lattice, with $n_U=n_V\ge2k$ so that no overlaps wrap around the ring, the partner
$u'=u+d$ shares $q=k-|d|$ neighbours for $1\le|d|\le k-1$. Hence
$\tilde m=\sum_{q=1}^{k-1}q(q-1)=k(k-1)(k-2)/3$ and
$P_A=\sum_{q=1}^{k-1}q(q-1)(q-2)=k(k-1)(k-2)(k-3)/4$. A pair of neighbours of $u$ at separation
$g$ is adjacent to $c=k-1-g$ further nodes, and there are $k-g$ such pairs, giving
$P_B=\sum_{h=1}^{k-1}h\binom{h-1}{2}=P_A/2$. The neighbour at offset $j\in\{0,\dots,k-1\}$ carries
$n_j=(k-1)^2-\tfrac12[j(j+1)+(k-1-j)(k-j)]$ squares, and elementary summation of
$\sum_j\binom{n_j}{2}-2P_A$ gives $P^{(2)}$. Substituting into~\eqref{eq:varQ}
gives~\eqref{eq:varlattice}, whose expansion in $1/k$ is
$81\kappa_2/(20k)+\bigl(\tfrac{27}{4}\kappa_1-\tfrac{87}{10}\kappa_2\bigr)k^{-2}+O(k^{-3})$. All four counts
were verified against direct enumeration for $3\le k\le15$.
\end{proof}

\begin{proof}[Proof of Theorem~\ref{thm:I2}]
With unit strengths $N_{uu'}=q$ and $M_{uu'}=\sum_{j=1}^{q}e^{i\psi_j}$, where
$\psi_j=\varphi_{uv_j}-\varphi_{u'v_j}$. The $2q$ link phases involved are distinct, so that the $\psi_j$ are independent and identically distributed with $\E[e^{i\psi}]=|\chi(1)|^{2}=r$, which is real. Then $\E[|M_{uu'}|^{2}]=\sum_{j,l}\E[e^{i(\psi_j-\psi_l)}]=q+q(q-1)r^{2}$, and division by $N_{uu'}^{2}=q^{2}$ gives $\E[(I^{(2)})^{2}]=q^{-1}+(1-q^{-1})r^{2}$, which is~\eqref{eq:I2law}. For $\varphi\sim\mathcal U[-\Delta,\Delta]$ we have $r=\sinc^{2}\Delta$.
\end{proof}

 \section{Relation to existing measures}\label{app:relations}

Setting all phases to zero in~\eqref{eq:C} gives
\[
  C_u=\frac{\sum_{u'\neq u}\sum_{v_1\neq v_2}
      |w_{uv_1}w_{v_1u'}w_{u'v_2}w_{v_2u}|}
     {\sum_{u'\neq u}\sum_{v_1\neq v_2}|w_{uv_1}w_{v_1u'}w_{u'v_2}|},
\]
the bipartite analogue of the weighted coefficient of Zhang and
Horvath~\cite{zhanghorvath2005}, an average of the closing link weighted by the strengths of the paths reaching it. For binary weights this becomes $m_u/|T(u)|$, the local form of the four-cycle to three-path ratio of Robins and Alexander~\cite{robins2004}. The coefficients of Lind \emph{et al.}~\cite{lind2005} and Zhang \emph{et al.}~\cite{zhang2008} differ only in the convention used to count potential squares and are monotonically related to it on the ensembles considered here.
 
The relation to the interferometric coefficient~\eqref{eq:krawciw} amounts to a change of motif. There the two interfering objects are a direct link and a two step path through an intermediary. A bipartite network possesses no direct link between nodes of the same class, and the two objects are instead two two-step paths through distinct intermediaries. Writing $P_1=w_{uv_1}w_{v_1u'}$ and $P_2=w_{uv_2}w_{v_2u'}$ we have $W_s=P_1\overline{P_2}$ and $\Rey W_s=|P_1||P_2|\cos(\arg P_1-\arg P_2)$, so that the numerator of~\eqref{eq:C} is a sum over interfering pairs of routes weighted by their amplitudes, which is the two-slit reading of Figure~\ref{fig:schematic}(c).
 
The closed and open measures constitute two views of one quantity. Summing $W_s$ over $v_1$ and $v_2$ at fixed $u$ and $u'$ gives $|M_{uu'}|^{2}$ by Proposition~\ref{prop:matrix}, and $|M_{uu'}|=q_{uu'}I^{(2)}_{uu'}$ for unit strengths, whence
\[
  \sum_{s\in S(u)}\Rey W_s
  =\sum_{u'\neq u}\Bigl(q_{uu'}^{2}\bigl(I^{(2)}_{uu'}\bigr)^{2}-q_{uu'}\Bigr).
\]
The clustering coefficient is therefore a degree-normalised aggregate of squared open-path visibilities.

\section{Computation, validation and robustness}\label{app:computation}

The accompanying code can be found at \url{https://github.com/anjaliwgh/complex_bipartite_networks.git}.

\subsection{Generation and cost}
Construction of the seed lattice costs $O(nk)$. Rewiring is carried out one node at a time, all links selected for rewiring at a node being removed together and replaced by a sample without replacement from that node's non-neighbours, which preserves the degrees in $\Uset$ exactly and avoids the bias introduced by link-by-link rewiring through collisions. Evaluation of the clustering coefficient through Proposition~\ref{prop:matrix} requires the sparse products $BB^{\mathsf H}$, $RR^{\mathsf T}$, $R^{\circ2}(R^{\circ2})^{\mathsf T}$ and $R(R^{\circ2})^{\mathsf T}$, each costing $O(nk^{2})$ operations and producing $O(nk^{2})$ nonzeros, followed by elementwise operations and row sums of the same order, giving $O(nk^{2})$ in time and memory. Explicit enumeration costs $O(n^{2}k)$. The counts entering
Theorem~\ref{thm:resolution} have the same cost. $P_A$ is a row sum over the nonzeros of $N$; $P_B$ requires, for each node, a lookup of $c_{v_1v_2}$ in the co-occurrence matrix $B^TB$ for each
of its $O(k^2)$ pairs of neighbours; and $n_{uv}=(NB)_{uv}-k_u-k_v+1$ on binary $B$ requires, for each link, a sum over the $k_v$ neighbours of $v$, $O(\sum_vk_v^2)=O(nk^2)$ in total. Hop-count distances require a breadth-first search from every node, at $O(n^{2}k)$. The apparent distance requires the propagator $(\mathbb I-\alpha A)^{-1}$, for which a dense inverse at $O(n^{3})$ is used, and this limits the system sizes reported here. For larger networks one may instead solve $(\mathbb I-\alpha A)x=e_y$ for a sample of sources at $O(nk)$ per solve, convergence being guaranteed by $\alpha<1/\rho(|A|)$.
 \subsection{Validation}

Four validation checks were performed before the production experiments. First, the matrix identities~\eqref{eq:num}--\eqref{eq:den} were compared with direct enumeration on random bipartite networks with heterogeneous strengths and random phases, agreeing to a precision of $3\times10^{-14}$. Second, setting all phases to zero reproduced the directly enumerated classical bipartite clustering coefficient to $10^{-10}$, with $Q_u=1$ at every node, as required by Proposition~\ref{prop:props}\,(iii). Third, random redefinitions of the phase references left $C_u$, $C_u^{(0)}$, $Q_u$ and $I^{(2)}_{uu'}$ unchanged to $10^{-16}$, confirming gauge invariance (Proposition~\ref{prop:props}\,(ii)). Finally, the measured mean phase factor agreed with the analytical prediction $(\sin\Delta/\Delta)^4$ within statistical error for every network size, degree and rewiring probability tested.
\subsection{Resolution}
Theorem~\ref{thm:resolution} was tested across the whole range of phase disorder. 

\emph{Covariances and counts.} The constants $\kappa_1$ and $\kappa_2$ of Lemma~\ref{lem:cov} were checked by direct sampling of the link phases, $4\times10^6$ samples per value, for each of the three kinds of overlapping pair at $\Delta\in \{\pi/4,\pi/2, 2, \pi\}$, and agree within sampling error. 
The counts $\tilde m_u$, $P^{(1)}_u$ and $P^{(2)}_u$, evaluated through the expressions given in the proof of Theorem~\ref{thm:resolution}, agree exactly with direct enumeration of all pairs of squares, on the lattice for $3\le k\le15$ and on rewired networks at $p\in\{0.1,0.5,1\}$ for nodes of both classes.

\emph{Regular lattice.} Figure~\ref{fig:resolution} compares Corollary~\ref{cor:lattice} with simulation on the lattice with $n_U=n_V=150$ and degrees $k\in\{3,4,6,10,14\}$, the degree controlling the number of squares available at each node. The measured standard deviation of $Q_u$ agrees with the prediction over the whole range of $\Delta$ and almost three orders of magnitude in $m_u$, whereas the estimate $2\sigma^{2}/m_u$ obtained by treating distinct squares as independent is low by up to a factor of $3.4$.

\emph{Rewired networks.} Since~\eqref{eq:varQ} holds for each node of any fixed topology, it was also tested in that form, on fixed networks with $n_U=n_V=100$, $k=6$ and $p\in\{0.1,1\}$, at every node of both classes lying on a square and at four values of $\Delta$. Over $4000$ phase realisations the measured variance of $Q_u$ matches the prediction to within the sampling error of a variance estimated from that many draws, with no trend in $p$. The independence estimate is low by a factor of $2.4$ at $p=0.1$ and $\Delta=\pi/4$, and by a factor of $1.1$ at $p=1$, where squares in a random bipartite network rarely share links.

\begin{figure}[t]
\centering
\includegraphics[width=\textwidth]{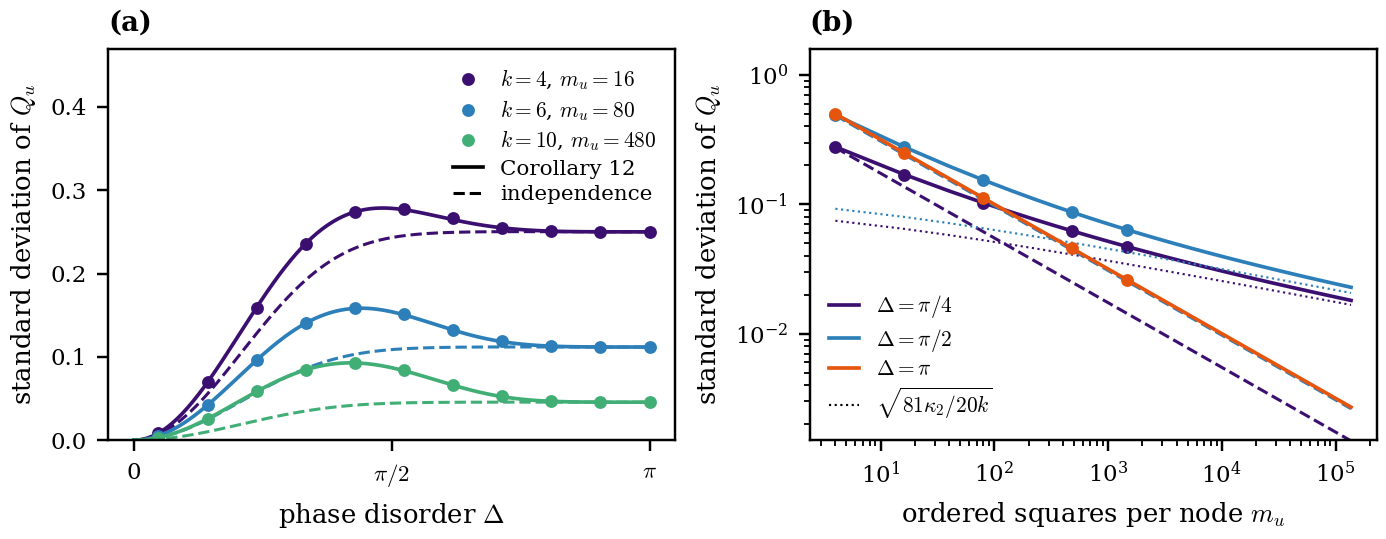}
\caption{The resolution of the phase factor is set by the number of squares only at complete phase disorder.
Regular ring lattice with $n_U=n_V=150$; $500$ realisations. Symbols denote simulations, solid lines the prediction of Corollary~\ref{cor:lattice}, and dashed lines the independent squares approximation $2\sigma^2/m_u$. \textbf{(a)} Standard deviation of $Q_u$ versus phase disorder for three degrees. The independence estimate is exact only at $\Delta=\pi$ and otherwise increasingly understates the fluctuations. \textbf{(b)} Standard deviation versus the number of ordered squares per node, for degrees $3$--$60$. At $\Delta=\pi$, $\mathrm{rms}(Q_u)=m_u^{-1/2}$ exactly; for $\Delta<\pi$, the data follow the asymptotic scaling $\sqrt{81\kappa_2/(20k)}$, corresponding to a decay of $m_u^{-1/6}$.}
\label{fig:resolution}
\end{figure}

\subsection{Attainable range}

Proposition~\ref{prop:range} was verified by direct minimisation of $Q_u$ over the link phases using repeated Nelder--Mead descent from random initial configurations. In all cases tested, the numerical minima matched the analytical bound within numerical tolerance. For complete bipartite networks, the minima are $-1$ for $K_{2,2}$, $-1/2$ for $K_{3,3}$, and $-1/3$ for $K_{4,4}$, in agreement with $-1/(q-1)$, where $q$ is the size of the opposite partition. The bipartite ring lattice attains $-1$ at degree $3$, where no pair of nodes shares more than two neighbours, and then $-0.625$ at degree $4$ and $-0.35$ at degree $6$ as denser overlaps appear. Random bipartite networks also agree with the predicted value to four decimal places. For the ensemble used in Section~\ref{sec:numerics}, at degree $6$, the coefficient therefore cannot fall below approximately $-0.35$ under any choice of phases, well outside the range explored by the uniform phase disorder considered there.
\subsection{Heterogeneous strengths}
Theorem~\ref{thm:meanphasefactor} is exact for unit strengths. Where strengths vary, the weighting correlates the numerator and denominator of $Q_u$ and introduces a correction of second order in the dispersion. Drawing strengths uniformly from $(0.5,1)$, the largest deviation of the measured phase factor from $\sinc^{4}\Delta$ over the whole range of $\Delta$ is $0.011$, so that the closed form remains an accurate approximation for realistic heterogeneity.
 
\subsection{Statistics and reproducibility}
Each parameter pair is averaged over forty realisations, and standard errors remain below two per cent of the mean for all clustering, distance and interference quantities. The exception is $\bar C$ at complete phase disorder, whose expectation vanishes so that the relative error is unbounded, and this value is reported but not interpreted. 
 
\end{document}